\documentclass[11pt]{article}
\usepackage[utf8]{inputenc}
\usepackage[margin=2.5cm]{geometry}

\usepackage{multirow, multicol}
\usepackage{booktabs}
\usepackage{dsfont,amssymb,amsmath,amsthm,centernot,graphicx,accents}
\usepackage{enumitem}
\usepackage{natbib}
\usepackage[colorlinks = true, citecolor = blue]{hyperref}
\usepackage{cleveref}
\usepackage{setspace}
\usepackage{comment}

\newtheorem*{theorem*}{Theorem}
\Crefname{theorem}{Theorem}{Theorems}
\crefname{theorem}{theorem}{theorems}

\newtheorem*{proposition*}{Proposition}
\Crefname{proposition}{Proposition}{Propositions}
\crefname{proposition}{proposition}{propositions}

\newtheorem*{corollary*}{Corollary}
\Crefname{corollary}{Corollary}{Corollaries}
\crefname{corollary}{corollary}{corollaries}

\newtheorem*{claim*}{Claim}
\Crefname{claim}{Claim}{Claims}
\crefname{claim}{claim}{claims}

\newtheorem{lemma}{Lemma}
\newtheorem*{lemma*}{Lemma}
\Crefname{lemma}{Lemma}{Lemmas}
\crefname{lemma}{lemma}{lemmas}

\newtheorem*{conjecture*}{Conjecture}
\Crefname{conjecture}{Conjecture}{Conjectures}
\crefname{conjecture}{conjecture}{conjectures}

\theoremstyle{definition}

\newtheorem*{definition*}{Definition}
\Crefname{definition}{Definition}{Definitions}
\crefname{definition}{definition}{definitions}

\newtheorem{assumption}{Assumption}
\newtheorem*{assumption*}{Assumption}
\Crefname{assumption}{Assumption}{Assumptions}
\crefname{assumption}{assumption}{assumptions}

\theoremstyle{definition}
\newtheorem{remark}{Remark}

\allowdisplaybreaks

\title{Valid Wald Inference with Many Weak Instruments}

\date{\today}
\author{Luther Yap}

\begin{document}

\onehalfspacing

\maketitle

\begin{abstract}
This paper proposes three novel test procedures that yield valid inference in an environment with many weak instrumental variables (MWIV). It is observed that the t statistic of the jackknife instrumental variable estimator (JIVE) has an asymptotic distribution that is identical to the two-stage-least squares (TSLS) t statistic in the just-identified environment. Consequently, test procedures that were valid for TSLS t are also valid for the JIVE t. Two such procedures, i.e., VtF and conditional Wald, are adapted directly. By exploiting a feature of MWIV environments, a third, more powerful, one-sided VtF-based test procedure can be obtained. 
\end{abstract}


\section{Introduction}

Consider an instrumental variable (IV) model where unit $i$ has outcome $Y_i$, scalar endogenous variable $X_i$ and a vector of instruments $Z_i$. With $\Pi_{i}=E\left[X_{i}\mid Z_{i}\right]$,
\begin{equation} \label{eqn:structural}
\begin{split}
Y_{i} & =\beta X_{i}+e_{i}\\
X_{i} & =\Pi_{i}+v_{i}
\end{split}
\end{equation}
where the object of interest is $\beta$. This paper is interested in the environment where $K:=dim(Z_i)$ is large and $cov(X_i, Z_i) \ne 0$ though the instrument is plausibly weak. Let $Z$ denote the $N\times K$ data matrix of instruments so $P:=Z\left(Z^{\prime}Z\right)^{-1}Z^{\prime}$ denotes the projection matrix, with $P_{ij}$ denoting the $(i,j)$th element of $P$. 
The concentration parameter 
$$S := \sum_{i}\sum_{j\ne i}P_{ij}\Pi_{i}\Pi_{j}/\sqrt{Var \left(\sum_{i}\sum_{j\ne i}P_{ij}X_{i}X_{j} \right) } $$ 
is a normalized object that characterizes how strong the instruments are. A setting with many weak instrument variables (MWIV) is one that is robust when $S \rightarrow \infty$ may not hold. 

This setting with MWIV is widely applicable. In this environment first proposed by \citet{bekker1994alternative}, the number of instrumental variables (IV) increases as the sample size increases. Such an environment can arise by design of the instrument. For instance, the judge instrument design uses indicators for the randomly assigned judges as instruments, so as the sample size increases the number of judges (and hence instruments) also increases. This design has applications in studying the effects of incarceration and detention (e.g., \citet{kling2006incarceration}; \citet{aizer2015juvenile}; \citet{dobbie2018effects}; \citet{bhuller2020incarceration}), the effect of bankruptcy (e.g, \citet{dobbie2015debt}), the effect of disability benefits (e.g., \citet{autor2019disability}); and the effect of foster care due to randomly assigned social workers (e.g., \citet{doyle2007child}). The shift-share design also generates many instruments by construction (e.g., \citet{adao2019shift}; \citet{goldsmith2020bartik}). The MWIV environment may also arise from constructed instruments. Authors may construct 150 instruments by interacting the quarter of birth with the state of birth (e.g., \citet{angrist1991does}). 

There is a large literature arguing that the standard two-stage-least-squares (TSLS) procedure for IV is biased and yields invalid inference in the many instruments environment (e.g., \citet{angrist1999jackknife}; \citet{donald2001choosing}; \citet{hansen2008estimation}). Consequently, a literature advocating for jackknife procedures has developed (e.g., \citet{angrist1999jackknife}; \citet{chao2012asymptotic}). 
In particular, the jackknife instrumental variable estimator (JIVE) proposed by \citet{angrist1999jackknife} (jiv2 in their paper) has gained much traction in empirical studies. This estimator, often referred to as the ``jackknife" or ``leave-out" estimator, has been used by papers including \citet{kling2006incarceration}, \citet{dobbie2018effects}, \citet{autor2019disability} and \citet{bhuller2020incarceration}. 
The JIVE estimator first constructs an instrument $\hat{X}_i := Z_i^\prime \hat{\pi}_{-i}$, where $\hat{\pi}_{-i}$ is the coefficient on $Z$ when regressing $X$ on $Z$ using all units except the $i$th unit. Then, using $\hat{X}$ to denote the data matrix of $\hat{X}_i$'s and $Y$ to denote the outcomes, $\hat{\beta}_{JIVE} = (\hat{X}^\prime X)^{-1} \hat{X}^\prime Y$, which essentially treats $\hat{X}$ as the instrument in a regular TSLS procedure. 

The popularity of JIVE can partially be attributed to its interpretability. Once $\hat{X}$ is viewed as a debiased constructed instrument, JIVE can be viewed as an extension of the familiar TSLS approach to the many instruments environment. 
Further, there are theoretical reasons to use JIVE. With a heteroskedastic model, \citet{chao2005consistent} showed how other estimators that are robust to many instruments, such as the limited information maximum likelihood (LIML) and the bias-corrected TSLS (BTSLS), require $S/ \sqrt{K} \rightarrow \infty$ for consistency. In contrast, \citet{chao2012asymptotic} showed that the JIVE remains consistent $S\rightarrow \infty$, so JIVE is more robust to weak instruments than BTSLS and LIML. Further, \citet{evdokimov2018inference} showed how, even with heterogeneous treatment effects, the JIVE estimand is a weighted average of treatment effects. 

While estimation using JIVE is justified, inference is a more difficult issue. The JIVE t-ratio $t^2 = (\hat{\beta}_{JIVE}-\beta_0)/ \widehat{var}(\hat{\beta}_{JIVE})$ is commonly used for testing the null hypothesis that $H_0: \beta = \beta_0$. However, there are two issues that arise with $t^2$. 
First, $\widehat{var}(\hat{\beta}_{JIVE})$ cannot be calculated in the same way as the variance estimator in TSLS. Since $\hat{X}$ is a constructed object, we cannot simply calculate the variance as if $\hat{X}$ were the instrument. Hence, a different variance estimator is required. Namely, building on \citet{chao2012asymptotic}, \citet{mikusheva2022inference} derived a variance estimator $\hat{V}$ that is consistent even under the alternative. 
The second issue concerns the critical value used for conducting the test. If $S \rightarrow \infty$, then comparing $t^2$ with $\chi^2$ (and hence $t$ with the standard normal) is justified. However, when $S$ does not diverge, $t^2$ has an asymptotic distribution that is non-standard. 
Consequently, a literature has developed on inference that is robust to both having many instruments and when the instruments are weak. These inference procedures include \citet{mikusheva2022inference} and \citet{crudu2021inference} that are based on the method of \citet{anderson1949estimation} adapted to MWIV, \citet{ayyar2022conditional} based on the conditional likelihood ratio, and \citet{matsushita2022jackknife} that is based on the LM statistic. 
These papers propose use test statistics that are different from $t^2$ so that their test statistics are asymptotically normal under the null. Hence, having an appropriate test procedure based on the JIVE t is still an open issue --- a gap that this paper aims to fill.  

Together with the development of MWIV procedures, there is also recent advancement in just-identified IV environments. The VtF procedure proposed in \citet{lmmpy23}, henceforth LMMPY, is found to be more powerful than many existing procedures, and can outperform many of them when considering lengths of confidence intervals (CI). Then, there is an open question on whether their novel method of curve construction applies to MWIV, and if so, whether these power properties remain. 

This paper first observes that the JIVE t in the MWIV environment has the same asymptotic distribution as the TSLS t statistic in the just-identified instrumental variable (IV) environment. This observation builds on \citet{mikusheva2022inference} who had an expression for the JIVE t when finding an analog of a first-stage screening procedure of \citet{StockYogo05}. 
The observation implies that any inference procedure valid for the TSLS t is also valid for the JIVE t. 
By defining the analogous terms appropriately, the conditional Wald procedure of \citet{moreira2003conditional} and the VtF procedure of LMMPY can both be implemented for the JIVE t. In particular, the same critical values can be used. 

Beyond the two adaptations of existing procedures into the new environment, this paper develops a third procedure that exploits a unique feature of the MWIV environment. In MWIV environments,the sign of $E[X_iX_j]$ for observations $i$ and $j$ with similar instrumental values is often known. For instance, in the judge environment where treatment (e.g., incarceration) is endogenous and the instrument (i.e., judges) are randomly assigned, observations $i$ and $j$ assigned to the same judge have $E[X_iX_j]\geq0$. By exploiting this information, this paper develops a third procedure that builds on VtF that is even more powerful. Notably, even if the assumption that $E[X_iX_j]\geq0$ does not hold, the proposed test is still valid --- it merely has less power.

The VtF-based approaches proposed in this paper retains the interpretability of the $t^2$ statistic based on JIVE. To implement the procedure after calculating $\hat{\beta}_{JIVE}$, practitioners merely have to use the appropriate variance estimator $\hat{V}$ to construct $t^2$ and the appropriate critical values, either from LMMPY or the one-sided values in this paper.

\section{Setting} \label{sec:setting}

Let $Y_{i},X_{i},Z_{i}$ denote the outcome, the endogenous variable, and the vector of instruments respectively. There are $N$ observations and $K$ instruments. The many $K$ instruments could arise from having $K$ different judges, for instance. Consider the linear IV model in \Cref{eqn:structural}. For some random variables $A$ and $B$, this paper uses the notation $Q_{AB} := (1/\sqrt{K}) \sum_{i=1}^{N}\sum_{j\ne i}P_{ij} A_i B_j$. In particular,

\[
\left(Q_{ee},Q_{Xe},Q_{XX}\right)^{\prime}:=\frac{1}{\sqrt{K}}\sum_{i=1}^{N}\sum_{j\ne i}P_{ij}\left(e_{i}e_{j},X_{i}e_{j},X_{i}X_{j}\right)^{\prime}
\]

Due to the setting, $e_i:=Y_i-X_i\beta$ is defined as the residual with respect to the true $\beta$. This paper makes a high-level assumption on the asymptotic distribution. Let $\mu^{2}:=\sum_{i}\sum_{j\ne i}P_{ij}\Pi_{i}\Pi_{j}$.
\begin{assumption} \label{asmp:distr}
\[
\left[\begin{array}{c}
Q_{ee}\\
Q_{Xe}\\
Q_{XX}-\frac{\mu^{2}}{\sqrt{K}}
\end{array}\right]\xrightarrow{d}N\left(0,\left[\begin{array}{ccc}
\Phi & \Sigma_{12} & \Sigma_{13}\\
\Sigma_{12} & \Psi & \tau\\
\Sigma_{13} & \tau & \Upsilon
\end{array}\right]\right)
\]
\end{assumption}

The asymptotic distribution of Assumption \ref{asmp:distr} is immediate from the structural model of \Cref{eqn:structural} and the central limit theorem of \citet{chao2012asymptotic} once some regularity conditions are satisfied. This paper makes a high level assumption on the distribution to abstract from the discussion of these regularity conditions. 

When doing a hypothesis test of $H_0:\beta=\beta_0$, the hypothesized $\beta_0$ is used instead. Let $\Delta:=\beta-\beta_{0}$ denote the divergence of the hypothesized value from the true value. Due to \citet{lim2022conditional}, Assumption \ref{asmp:distr} implies that, for $e_{i}\left(\beta_{0}\right):=Y_{i}-\beta_{0}X_{i}$, 

\[
\left[\begin{array}{c}
Q_{e\left(\beta_{0}\right)e\left(\beta_{0}\right)} - \Delta^2 \frac{\mu^{2}}{\sqrt{K}} \\
Q_{Xe\left(\beta_{0}\right)} - \Delta \frac{\mu^{2}}{\sqrt{K}}\\
Q_{XX} - \frac{\mu^{2}}{\sqrt{K}}
\end{array}\right]\xrightarrow{d}N\left(\left[\begin{array}{c}
0\\
0\\
0
\end{array}\right],\left[\begin{array}{ccc}
\Phi\left(\beta_{0}\right) & \Sigma_{12}\left(\beta_{0}\right) & \Sigma_{13}\left(\beta_{0}\right)\\
\Sigma_{12}\left(\beta_{0}\right) & \Psi\left(\beta_{0}\right) & \tau\left(\beta_{0}\right)\\
\Sigma_{13}\left(\beta_{0}\right) & \tau\left(\beta_{0}\right) & \Upsilon
\end{array}\right]\right)
\]

where, 
\begin{align*}
\Phi\left(\beta_{0}\right) & =\Delta^{4}\Upsilon+4\Delta^{3}\tau+\Delta^{2}\left(4\Psi+2\Sigma_{13}\right)+4\Delta\Sigma_{12}+\Phi\\
\Sigma_{12}\left(\beta_{0}\right) & =\Delta^{3}\Upsilon+3\Delta^{2}\tau+\Delta\left(2\Psi+\Sigma_{13}\right)+\Sigma_{12}\\
\Sigma_{13}\left(\beta_{0}\right) & =\Delta^{2}\Upsilon+2\Delta\tau+\Sigma_{13}\\
\Psi\left(\beta_{0}\right) & =\Delta^{2}\Upsilon+2\Delta\tau+\Psi\\
\tau\left(\beta_{0}\right) & =\Delta\Upsilon+\tau
\end{align*}

Then, defining $S:=\mu^2/\sqrt{K\Upsilon}$ as the concentration parameter (which is numerically equivalent to how $S$ was written in the introduction), 

\[
Q\left(\beta_{0}\right):=\left[\begin{array}{c}
AR\left(\beta_{0}\right)\\
\xi\left(\beta_{0}\right)\\
\nu
\end{array}\right] := \left[\begin{array}{c}
Q_{e\left(\beta_{0}\right)e\left(\beta_{0}\right)}/\sqrt{\Phi\left(\beta_{0}\right)}\\
Q_{Xe\left(\beta_{0}\right)}/\sqrt{\Psi\left(\beta_{0}\right)}\\
Q_{XX}/\sqrt{\Upsilon}
\end{array}\right]\xrightarrow{d}N\left(\left[\begin{array}{c}
\Delta^{2}S\sqrt{\frac{\Upsilon}{\Phi\left(\beta_{0}\right)}}\\
\Delta S\sqrt{\frac{\Upsilon}{\Psi\left(\beta_{0}\right)}}\\
S
\end{array}\right],V(Q(\beta_0))\right)
\]
where
\begin{align*}
    V(Q(\beta_0)) := \left[\begin{array}{ccc}
1 & \Sigma_{12}\left(\beta_{0}\right)/\sqrt{\Phi\left(\beta_{0}\right)\Psi\left(\beta_{0}\right)} & \Sigma_{13}\left(\beta_{0}\right)/\sqrt{\Phi\left(\beta_{0}\right)\Upsilon}\\
\Sigma_{12}\left(\beta_{0}\right)/\sqrt{\Phi\left(\beta_{0}\right)\Psi\left(\beta_{0}\right)} & 1 & \tau\left(\beta_{0}\right)/\sqrt{\Psi\left(\beta_{0}\right)\Upsilon}\\
\Sigma_{13}\left(\beta_{0}\right)/\sqrt{\Phi\left(\beta_{0}\right)\Upsilon} & \tau\left(\beta_{0}\right)/\sqrt{\Psi\left(\beta_{0}\right)\Upsilon} & 1
\end{array}\right]
\end{align*}

In the MWIV environment, the concentration parameter $S$ does not diverge asymptotically. Conversely, in an environment with ``strong" instruments, $S\rightarrow \infty$. In light of the above asymptotic distribution, and how the variance objects can be consistently estimated, the normalized statistics in $Q(\beta_0)$, i.e., $(AR(\beta_0),\xi(\beta_0),\nu)$, have been used for inference on $\beta$. In particular, \citet{mikusheva2022inference} used the $AR(\beta_{0})$ statistic, while \citet{matsushita2022jackknife} used the $\xi(\beta_{0})$ statistic as their LM procedure. Under the null, both test statistics are normally distributed, so an appropriate critical value from the standard normal distribution can be used. Since $\Delta=0$ under the null, characterizing the alternative distribution is unnecessary for developing a valid test. However, characterizing such a distribution for some $\beta_0$ in general is helpful for power comparisons. 

In this characterization, although $Q_{e(\beta_0)e(\beta_0)}, Q_{Xe(\beta_0)}, Q_{XX}$ can be obtained immediately from the data, objects in the variance, such as $\Upsilon,\Psi(\beta_0)$ and $\Phi(\beta_0)$, cannot be feasibly obtained. Hence, $Q(\beta_0)$ is treated as an asymptotic object that is not feasible. Nonetheless, there are feasible estimators for the variance objects that are consistent even under the alternative. Following \citet{mikusheva2022inference}, let:
\begin{align*}
M & :=I-P\\
\tilde{P}_{ij}^{2} & :=\frac{P_{ij}^{2}}{M_{ii}M_{jj}+M_{ij}^{2}}
\end{align*}
so $M$ is the annihilator matrix, and $\tilde{P}_{ij}$ adjusts the $P_{ij}$ object. Then, the feasible and consistent variance estimators are:
\begin{align*}
\hat{\Upsilon} & :=\frac{1}{K}\sum_{i}\left(\sum_{j\ne i}P_{ij}X_{j}\right)^{2}\frac{X_{i}M_{i}X}{M_{ii}}+\frac{1}{K}\sum_{i}\sum_{j\ne i}\tilde{P}_{ij}^{2}M_{i}XX_{i}M_{j}XX_{j}\\
\hat{\tau}\left(\beta_{0}\right) & :=\frac{1}{2}[\frac{1}{K}\sum_{i}\left(\sum_{j\ne i}P_{ij}X_{j}\right)^{2} \left(\frac{X_{i}M_{i}e\left(\beta_{0}\right)}{M_{ii}} +\frac{e_{i}\left(\beta_{0}\right)M_{i}X}{M_{ii}}\right) \\
&\qquad+\frac{1}{K}\sum_{i}\sum_{j\ne i}\tilde{P}_{ij}^{2}(M_{i}XX_{i}M_{j}Xe_{j}\left(\beta_{0}\right)+M_{i}Xe_{i}\left(\beta_{0}\right)M_{j}XX_{j})]\\
\hat{\Psi}\left(\beta_{0}\right) & :=\frac{1}{K}\sum_{i}\left(\sum_{j\ne i}P_{ij}X_{j}\right)^{2}\frac{e_{i}\left(\beta_{0}\right)M_{i}e\left(\beta_{0}\right)}{M_{ii}}+\frac{1}{K}\sum_{i}\sum_{j\ne i}\tilde{P}_{ij}^{2}M_{i}Xe_{i}\left(\beta_{0}\right)M_{j}Xe_{j}\left(\beta_{0}\right)
\end{align*}

Define the normalized statistics as $\hat{\xi}(\beta_0) := Q_{Xe(\beta_0)}/\sqrt{\Psi(\beta_0)}$, $\hat{\nu} := Q_{XX}/\sqrt{\hat{\Upsilon}}$, and $\hat{\rho}(\beta_0) := \hat{\tau}(\beta_{0})/\sqrt{\hat{\Psi}\left(\beta_{0}\right)\hat{\Upsilon}}$, where $\rho\left(\beta_{0}\right):=\tau\left(\beta_{0}\right)/\sqrt{\Psi\left(\beta_{0}\right)\Upsilon}$. I make a high-level assumption on these variance estimators that they converge to the true variance objects. The primitives of the assumption can be justified by \citet{mikusheva2022inference}, who similarly used analogs of these objects.

\begin{assumption} \label{asmp:var_consistent}
$\hat{\Psi}(\beta_0) \xrightarrow{p} \Psi(\beta_0)$, $\hat{\Upsilon} \xrightarrow{p} \Upsilon$ and $\hat{\tau}(\beta_0) \xrightarrow{p} \tau(\beta_0)$. 
\end{assumption}
A corollary of of the assumption is that $\hat{\rho} (\beta_0) \xrightarrow{p} \rho (\beta_0)$ due to the continuous mapping theorem. 

As a special case of the setup, we can consider the judge design. In particular, in the judge design without covariates, it can be shown that $E\left[\nu\right]\geq0$. In this environment, $Z_{i}$ are judge indicators. Let $k(i)$ denote the judge $k$ that $i$ is matched to, and $\pi$ denote the vector of values that $E\left[X_{i}|Z_{i}\right]$ can take for the respective judges, i.e., $\pi_{k}=E\left[X_{i}|k(i)=k\right]$. 

\[
X_{i}=Z_{i}^{\prime}\pi+v_{i}
\]

Since $P_{ij}=1\left\{ k(i)=k(j)\right\} /N_{k(i)}$, where $N_{k}$ is the number of observations assigned to judge $k$,
\[
Q_{XX}=\frac{1}{\sqrt{K}}\sum_{i=1}^{N}\sum_{j\ne i}\frac{1\left\{ k(i)=k(j)\right\} }{N_{k(i)}}X_{i}X_{j}
\]

For any two individuals $i,j$ matched to the same judge, we have:
\begin{align*}
E\left[X_{i}X_{j}\right] & =E\left[\left(\pi_{k}+v_{i}\right)\left(\pi_{k}+v_{j}\right)\right]=\pi_{k}^{2}\geq0
\end{align*}
Hence, $S=E[\nu]\geq 0$. This observation is helpful in motivating the test procedure.  

\section{Procedure}

\subsection{JIVE t statistic} \label{sec:jive_t}
An existing procedure that is robust to having many instruments (albeit not when they are weak) is the jackknife instrumental variables estimator (JIVE). In particular, using a leave-one-out approach to eliminate bias, \citet{angrist1999jackknife} proposes using 
\begin{equation} \label{eqn:beta_jive}
    \hat{\beta}_{JIVE}=\frac{\sum_{i}\sum_{j\ne i}P_{ij}Y_{i}X_{j}}{\sum_{i}\sum_{j\ne i}P_{ij}X_{i}X_{j}}
\end{equation}

The estimator described in the introduction is numerically equivalent to the expression in \Cref{eqn:beta_jive}. To obtain a variance estimator for $\hat{\beta}_{JIVE}$ that is consistent, \citet{chao2012asymptotic} proposed an analogous jackknife approach, which \citet{mikusheva2022inference} subsequently refined. Using $\hat{e}_i := Y_i - X_i \hat{\beta}_{JIVE}$ to denote the JIVE residual, the variance estimator is given by:
\begin{align*}
\hat{V}= & \frac{\sum_{i}\left(\sum_{j\ne i}P_{ij}X_{j}\right)^{2}\frac{\hat{e}_{i}M_{i}\hat{e}}{M_{ii}}+\sum_{i}\sum_{j\ne i}\tilde{P}_{ij}^{2}M_{i}X\hat{e}_{i}M_{j}X\hat{e}_{j}}{\left(\sum_{i}\sum_{j\ne i}P_{ij}X_{i}X_{j}\right)^{2}}
\end{align*}

Then, the JIVE inference procedure proposed in the literature uses the following t-statistic:
\begin{align*}
\hat{t}_{JIVE}^2 & =\frac{\left(\hat{\beta}_{JIVE}-\beta_{0}\right)^{2}}{\hat{V}}
\end{align*}

As \citet{chao2012asymptotic} have shown, if the concentration parameter diverges asymptotically i.e., $S\rightarrow \infty$, then $\hat{t}_{JIVE}$ will have a standard normal distribution asymptotically, so using $\pm 1.96$ critical values for $\hat{t}_{JIVE}$ will be valid. With many instruments, $S\rightarrow \infty$ is analogous to having a strong instrument. But when we have many \emph{weak} instruments, $S$ does not diverge, so using the $\pm 1.96$ critical values for $\hat{t}_{JIVE}$ will not result in valid inference. 

Nonetheless, $\hat{t}_{JIVE}$ can algebraically be expressed as a function of $\hat{\xi}(\beta_0),\hat{\rho}(\beta_0)$ and $\hat{\nu}$, which are feasible normalized statistics defined in Section 2. Since the variance estimators used in these feasible objects are consistent, these statistics converge to $\xi(\beta_0),\rho(\beta_0)$ and $\nu$ respectively, which have a known joint distribution under the asymptotic environment. 

\begin{lemma} \label{lem:JIVE}
Under Assumption \ref{asmp:var_consistent},
\begin{align*}
\hat{t}_{JIVE}^2 & =\frac{\hat{\xi}\left(\beta_{0}\right)^{2}}{1-2\frac{\hat{\xi}\left(\beta_{0}\right)}{\hat{\nu}}\hat{\rho}\left(\beta_{0}\right)+\frac{\hat{\xi}\left(\beta_{0}\right)^{2}}{\hat{\nu}^{2}}}\\
 & =\frac{\xi\left(\beta_{0}\right)^{2}(1+o_{P}(1))}{1-2\frac{\xi\left(\beta_{0}\right)}{\nu}\rho\left(\beta_{0}\right)+\frac{\xi\left(\beta_{0}\right)^{2}}{\nu^{2}}} =: t^2_{JIVE} (1+o_{P}(1))
\end{align*}
\end{lemma}

This lemma builds on \citet{mikusheva2022inference}, who had a similar expression in their theorem 5. First, Lemma \ref{lem:JIVE} shows a numerical equivalence between the various objects that can be feasibly calculated, rather than just an asymptotic result. Second, the asymptotic result in Lemma \ref{lem:JIVE} holds not just under the null, but also for any $\beta_0 \ne \beta$, and this result is immediate from Assumption \ref{asmp:var_consistent}.

By inspection, the $t_{JIVE}$ statistic has the same distribution as the two-stage-least-squares $t_{TSLS}$ for just-identified IV setup considered in LMMPY. In particular, $\rho(\beta_0)$ has a similar interpretation in both environments. $\nu$ here takes the role of their $f$ and $\xi(\beta_0)$ here takes the role of their $t_{AR}(\beta_0)$. Finally, $S$ in MWIV takes the place of their $f_0$. Hence, when proposing a \citet{StockYogo05} analog, \citet{mikusheva2022inference} used $\nu^2$ as the analog of the $F$-statistic of the first-stage regression in the just-identified environment. 

Since the asymptotic distribution of $\hat{t}_{JIVE}^2$ is known under the null, it is possible to construct a critical value function for $\hat{t}_{JIVE}^2$ that is robust to MWIV. The VtF-based critical value function $c(.)$ proposed in this paper is only a function of $\nu$ and $\rho(\beta_0)$, which can be consistently estimated from the data. Then, for testing $H_0:\beta=\beta_0$, reject if $\hat{t}^2_{JIVE} \geq c(\nu,\rho(\beta_0))$. For all $S,\rho$ and a size $\alpha$ test, the critical value function will satisfy:
\begin{align*}
    \Pr(t_{JIVE}^2>c(\nu,\rho(\beta_0))) = \alpha
\end{align*}

\subsection{Description of curve construction}

There are (at least) two possible ways to construct the curve. The first way is a two-sided construction, which is identical to the VtF procedure in LMMPY. Since $t_{JIVE}^2$ has the same asymptotic distribution as $t_{TSLS}^2$, any inference procedure that is valid for the just-identified $t_{TSLS}^2$ is also valid in this context. Hence, if we are agnostic about $S=E\left[\nu\right]$, we can immediately use the critical values calculated in LMMPY. 

The second way is a one-sided construction, which I will call the one-sided VtF (VtFo). In the judge environment with many instruments, as seen in \Cref{sec:setting}, an argument from the setting shows that there may be good reason to believe that $S\geq0$. By using this piece of information in the curve construction, we may be able to get a more powerful test. The cost is that whenever we observe $\nu<0$, we cannot reject the null, and must conclude that the data is uninformative. The assumption of $S\geq0$ is also a more natural way to think about instruments in the just-identified environment, because researchers often justify the relevance condition by arguing that the instrument $Z$ affects the endogenous variable $X$ in a particular direction. Construction of a curve assuming $S \geq 0$ in the just-identified environment may hence also be of independent interest. 

Considering how the $AR(\beta_0)$ statistic is asymptotically normal and, for a 5\% test, \citet{mikusheva2022inference} proposed using 1.645 instead of $\pm 1.96$ as the critical value for $AR(\beta_0)$, VtFo is essential for comparing the recently-developed VtF with the existing literature on MWIV. 

The remainder of this subsection outlines the one-sided VtF critical value curve construction, while details are relegated to \Cref{sec:curve_details}. The method of construction is similar to the VtF in LMMPY, though the recursive equations differ. Since the curve is constructed under the null, I drop the $\beta_0$ indices without risk of ambiguity. Let $T:=\nu-\rho\xi$, mimicking the $Q$ in LMMPY. If $\rho\ne0$, then $\xi=\frac{1}{\rho}\left(\nu-T\right)$, which means

\begin{align*}
t^{2} & =\frac{\xi^{2}}{1-2\rho\frac{\xi}{\nu}+\frac{\xi^{2}}{\nu^{2}}}\\
 & =\frac{\frac{\nu^{2}}{\rho^{2}}\left(\nu-T\right)^{2}}{\nu^{2}-2\rho\nu\left(\nu-T\right)\frac{1}{\rho}+\frac{\left(\nu-T\right)^{2}}{\rho^{2}}}\\
 & =\frac{\nu^{2}\left(\nu-T\right)^{2}}{\rho^{2}T^{2}+\left(1-\rho^{2}\right)\left(\nu-T\right)^{2}}
\end{align*}

For notational compactness, I use:
\begin{equation} \label{eqn:t_fn}
    t^2(\nu,T,\rho) := \frac{\nu^{2}\left(\nu-T\right)^{2}}{\rho^{2}T^{2}+\left(1-\rho^{2}\right)\left(\nu-T\right)^{2}}
\end{equation}
As in LMMPY, for a given $T$ and $\rho$, the plot of $t^2(\nu,T,\rho)$ against $\nu$ will be W-shaped. Suppose for now that $\left\{ \nu:t^{2}\leq c(\nu,\rho)\right\} =(-\infty,\bar{\nu}]$ for some $\bar{\nu}$ given $T$ and $\rho$. Since $\nu|T\sim N(T,\rho^{2})$, for a test of correct size $\alpha$, we must have:

\begin{equation} \label{eqn:sys1}
\frac{\bar{\nu}^{2}\left(\bar{\nu}-T\right)^{2}}{\rho^{2}T^{2}+\left(1-\rho^{2}\right)\left(\bar{\nu}-T\right)^{2}}-c\left(\bar{\nu},\rho\right) =0
\end{equation}

\begin{equation} \label{eqn:sys2}
\Phi\left(\frac{\bar{\nu}-T}{\sqrt{\rho^{2}}}\right)  =1-\alpha
\end{equation}
where, with some abuse of notation, $\Phi(.)$ is the normal CDF. By solving for $T$ in \Cref{eqn:sys2}, we can substitute $T$ into \Cref{eqn:sys1} to initialize the curve with a closed form solution. Where $\sqrt{q}:=\Phi^{-1}(1-\alpha)$,
\begin{equation} \label{eqn:closed_c}
c\left(\bar{\nu},\rho\right)=\frac{\bar{\nu}^{2}}{\rho^{2}\left(\frac{\bar{\nu}}{\sqrt{\rho^{2}}\sqrt{q}}-1\right)^{2}+\left(1-\rho^{2}\right)}
\end{equation}

The curve can be constructed from the fixed point $\nu^{*}=\sqrt{\rho^{2}q}$, $c^{*}=\frac{\rho^{2}q}{1-\rho^{2}}$.\footnote{To see that this is the fixed point, \Cref{eqn:sys2} yields $\frac{\nu^{*}-T}{\sqrt{\rho^{2}}}=\Phi^{-1}\left(1-\alpha\right)=\sqrt{q}$ so $T=0$, $\nu^{*}=\sqrt{\rho^{2}q}$. Further, $\frac{\nu^{2}\left(\nu-T\right)^{2}}{\rho^{2}T^{2}+\left(1-\rho^{2}\right)\left(\nu-T\right)^{2}}=\frac{\rho^{2}q}{1-\rho^{2}}$. Hence, by using $\nu^{*}=\sqrt{\rho^{2}q},c^{*}=\frac{\rho^{2}q}{1-\rho^{2}}$ in the mapping, we get back the same point, giving us the fixed point as before. } 

Observe that \Cref{eqn:sys1} is quartic in $\bar{\nu}$ in the numerator, so it has a local maxima. As $T$ gets larger, the middle local maxima gets larger, and eventually crosses the critical value function $c(.)$, so that the initial assumption of $\left\{ \nu:t^{2}\leq c(\nu,\rho)\right\} =(-\infty,\bar{\nu}\left(T,\rho\right)]$ no longer holds. This situation means that there is some $\nu<\bar{\nu}$ such that $\frac{\nu^{2}\left(\nu-T\right)^{2}}{\rho^{2}T^{2}+\left(1-\rho^{2}\right)\left(\nu-T\right)^{2}}>c\left(\bar{\nu},\rho\right)$. Then, the critical value function for such $T$ values must satisfy the following four equations: 

\begin{equation} \label{eqn:3cross_nuH}
\frac{\nu_{H}^{2}\left(\nu_{H}-T\right)^{2}}{\rho^{2}T^{2}+\left(1-\rho^{2}\right)\left(\nu_{H}-T\right)^{2}}-c\left(\nu_{H},\rho\right) =0
\end{equation}
\begin{equation} \label{eqn:3cross_nuM}
\frac{\nu_{M}^{2}\left(\nu_{M}-T\right)^{2}}{\rho^{2}T^{2}+\left(1-\rho^{2}\right)\left(\nu_{M}-T\right)^{2}}-c\left(\nu_{M},\rho\right) =0
\end{equation}
\begin{equation} \label{eqn:3cross_nuL}
\frac{\nu_{L}^{2}\left(\nu_{L}-T\right)^{2}}{\rho^{2}T^{2}+\left(1-\rho^{2}\right)\left(\nu_{L}-T\right)^{2}}-c\left(\nu_{L},\rho\right) =0
\end{equation}
\begin{equation} \label{eqn:3cross_probability}
\Phi\left(\frac{\nu_{H}-T}{\sqrt{\rho^{2}}}\right)-\Phi\left(\frac{\nu_{M}-T}{\sqrt{\rho^{2}}}\right)+\Phi\left(\frac{\nu_{L}-T}{\sqrt{\rho^{2}}}\right) =1-\alpha
\end{equation}

where $\nu_L \leq \nu_M \leq \nu_H$. We do not have a neat closed-form solution as before, but we can construct the curve iteratively, thereby tracing out the rest of the curve. Details are provided in \Cref{sec:curve_details}, which readers may wish to skip on the first reading. The critical value functions are displayed in \Cref{fig:cvf_curves}.

\begin{figure}
    \centering
    \caption{Plot of Critical Value Functions for various $\rho$ values} \label{fig:cvf_curves} 
    \includegraphics[scale=.2]{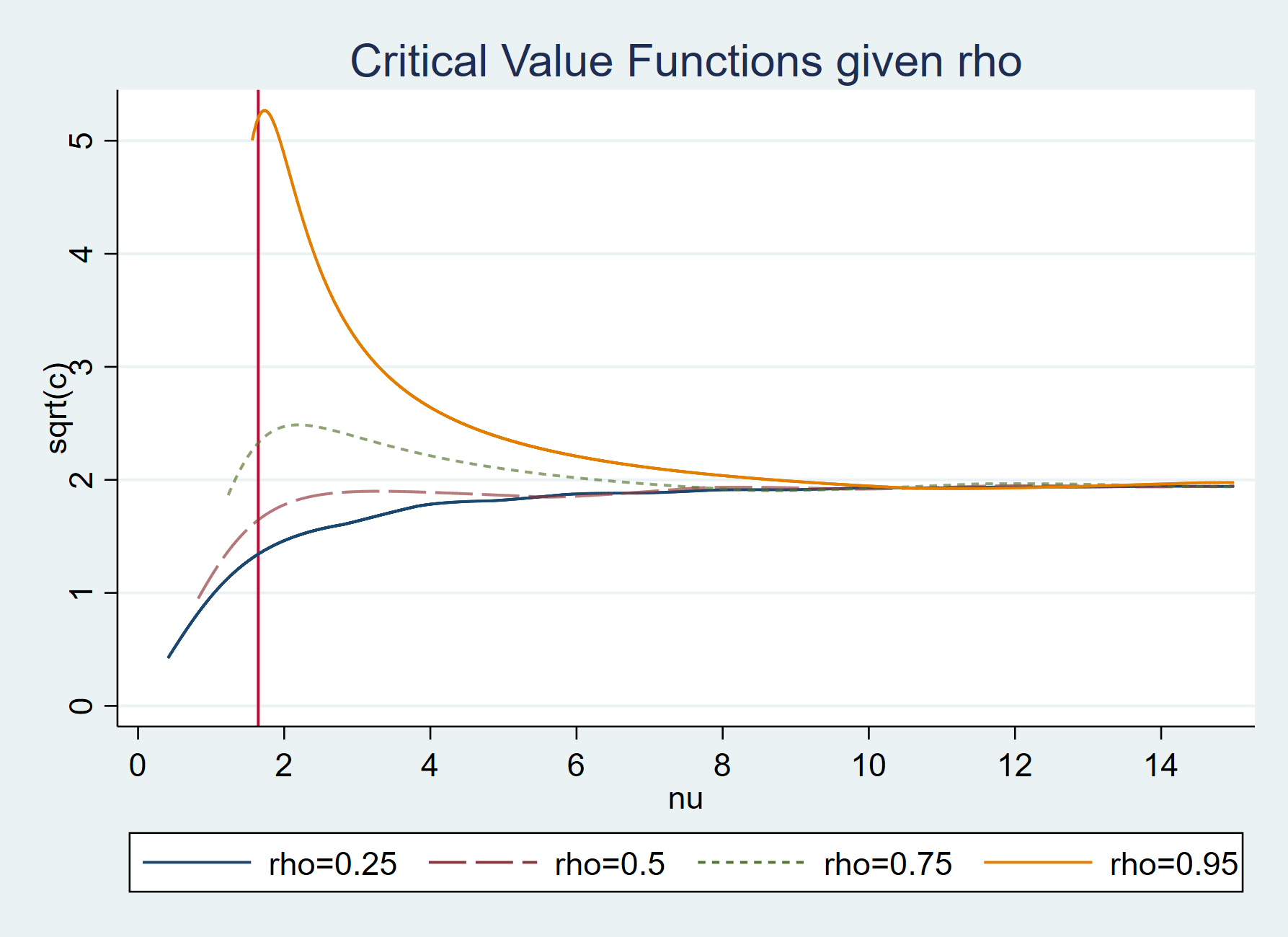} 
    \footnotesize\flushleft The vertical line corresponds to $\nu = 1.645$. The plot displays critical values for a 5\% test. 
\end{figure}

Many features of this VtFo curves are similar to the VtF curves of LMMPY. For any value of $\rho$, as $\nu$ increases, $\sqrt{c}$ converges to the standard normal critical value of 1.96. 
Even though this test is motivated by a one-sided first-stage, we are still conducting a two-sided test using the $t_{JIVE}$ statistic, so 1.96 is the relevant limit.
When there is a low degree of endogeneity (i.e., $\rho$ is small), $\sqrt{c}$ lies below 1.96, so using the conventional $\pm1.96$ factors will be conservative. For any value with $|\rho| < 1$, VtFo may reject the null even when $\nu < 1.645$, because the critical value function is still defined as long as $\nu > \rho \Phi^{-1}(1-\alpha)$.

\begin{remark}
Since any inference procedure that is valid for $t_{TSLS}$ is also valid for $t_{JIVE}$, the Conditional Wald (CW) approach of \citet{moreira2003conditional} is also valid in the many instruments environment. In particular, using a test where we reject if $\hat{t}_{JIVE}^{2}>c_{CW}\left(\hat{\rho}({\beta_{0}}),\hat{T}\right)$ is valid, where $c_{CW}$ is the same critical value function as in \citet{moreira2003conditional}. 
\end{remark}

\subsection{Details of One-Sided Curve Construction} \label{sec:curve_details}
To be precise about how the procedure works, with a given $\rho$, we can use the expression in \Cref{eqn:closed_c} up to a value $\tilde{\nu}$. Then, we have to switch to the three-intersection algorithm. The remainder of this subsection first describes how to find $\tilde{\nu}$, then describes the three-intersection algorithm.

To find $\tilde{\nu}$, first fix $\rho$ and construct a grid of $T>0$. For every $T$, construct a grid of $\nu$ to calculate $t^2(\nu,T,\rho)$ using \Cref{eqn:t_fn} and $c(\nu,\rho)$ using \Cref{eqn:closed_c}. 
For small values of $T$, the curves of $t^2(\nu,T,\rho)$ and $c(\nu,\rho)$ are displayed in Figure \ref{fig:1cross}, where $t^2$ is the W-curve in solid line while $c(\nu,\rho)$ is the dashed line. Since the rejection rule is $t^2(\nu,T,\rho) > c(\nu,\rho)$, for a given $T$ and $\rho$, to find the non-rejection probability, we must integrate over the region of $\nu$ such that the dashed line lies above the solid line. Hence, $\left\{ \nu:t^{2}\leq c(\nu,\rho)\right\} =(-\infty,\bar{\nu}]$, where $\bar{\nu}$ is the $\nu$ value where the two curves intersect. Thus, the non-rejection probability is given by $\Phi((\bar{\nu}-T)/\sqrt{\rho^2})$, motivating \Cref{eqn:sys2}. Since the curves intersect, $t^2(\bar{\nu},T,\rho) = c(\bar{\nu},\rho)$ at this $\bar{\nu}$ value, motivating \Cref{eqn:sys1}. For large values of $T$, as seen in \Cref{fig:3cross}, $c(\nu,\rho)$ and $t^2(\nu,T,\rho)$ cross multiple times. Since we had a grid of $T$, we can find $\tilde{T}$, the first $T$ value such that there are multiple crossings. This $\tilde{T}$ corresponds to the value where the two curves are tangent to each other at some $\nu$, as seen in \Cref{fig:tancross}. With a given $\tilde{T}$, we can use \Cref{eqn:sys1} to solve for the corresponding $\tilde{\nu}$. 

When there are three intersections (as in \Cref{fig:3cross}), we can create a grid of $T$ values starting from $\tilde{T}$ where the three intersection issue first occurs. The system of equations (\ref{eqn:3cross_nuH}) to (\ref{eqn:3cross_probability}) is motivated by the rejection rule of $t^2(\nu,T,\rho) > c(\nu,\rho)$, so we want to integrate over the $\nu$ region where the solid line is below the dashed line. Starting from the smallest $T$ in this grid, we know that $\nu_L \leq \nu_M \leq T \leq \nu_H$. Hence, $c(\nu_M,\rho)$ and $c(\nu_L,\rho)$ are known due to \Cref{eqn:closed_c}. Solve \Cref{eqn:3cross_nuM} and \Cref{eqn:3cross_nuL} to obtain a pair $\nu_L,\nu_M$ corresponding to the given $T$. Then, use $\nu_L,\nu_M$ for the given $T$ in \Cref{eqn:3cross_probability} to obtain a $\nu_H > T$. Use this $\nu_H$ in \Cref{eqn:3cross_nuH} to obtain a $c(\nu_H,\rho)$ value. Repeat this process for subsequent $T$'s in the sequence. The $(\nu_H, c(\nu_H,\rho))$ pair obtained in the first iteration will eventually become a corresponding $(\nu_M, c(\nu_M,\rho))$ pair for some $T$ that is large enough. This iterative procedure allows us to extend the curve indefinitely.

\begin{figure}
    \centering
    \caption{Curves with one crossing} \label{fig:1cross}
    \includegraphics[scale=.2]{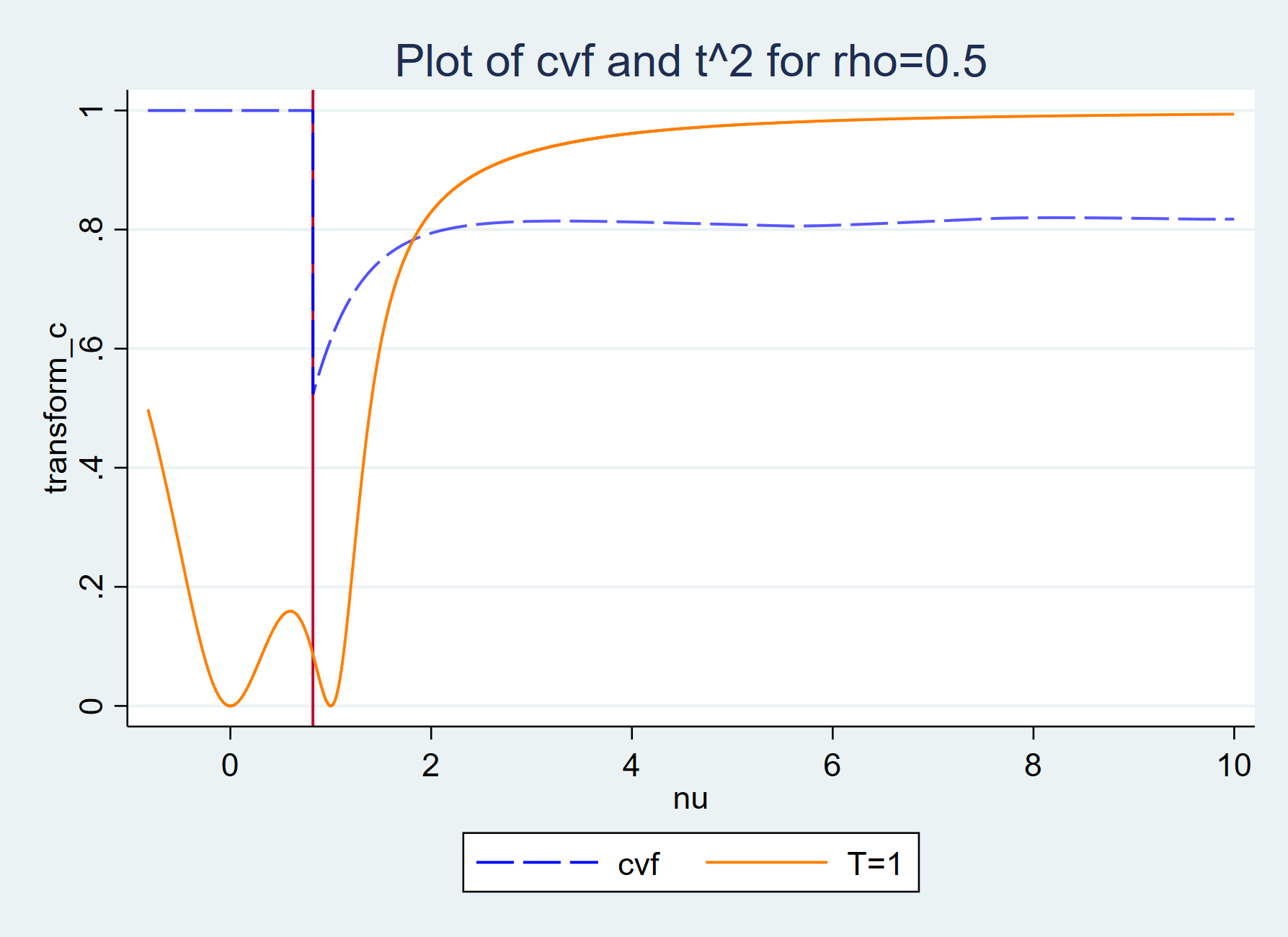} 
    \footnotesize\flushleft The dashed line corresponds to $c(\nu,\rho)$ while the solid line corresponds to $t^2(\nu,T,\rho)$. The $y$ axis has been transformed using $y^* = \frac{y}{1+y}$, so $\infty$ corresponds to 1.
\end{figure}

\begin{figure}
    \centering
    \caption{Curves with three crossings} \label{fig:3cross}
    \includegraphics[scale=.2]{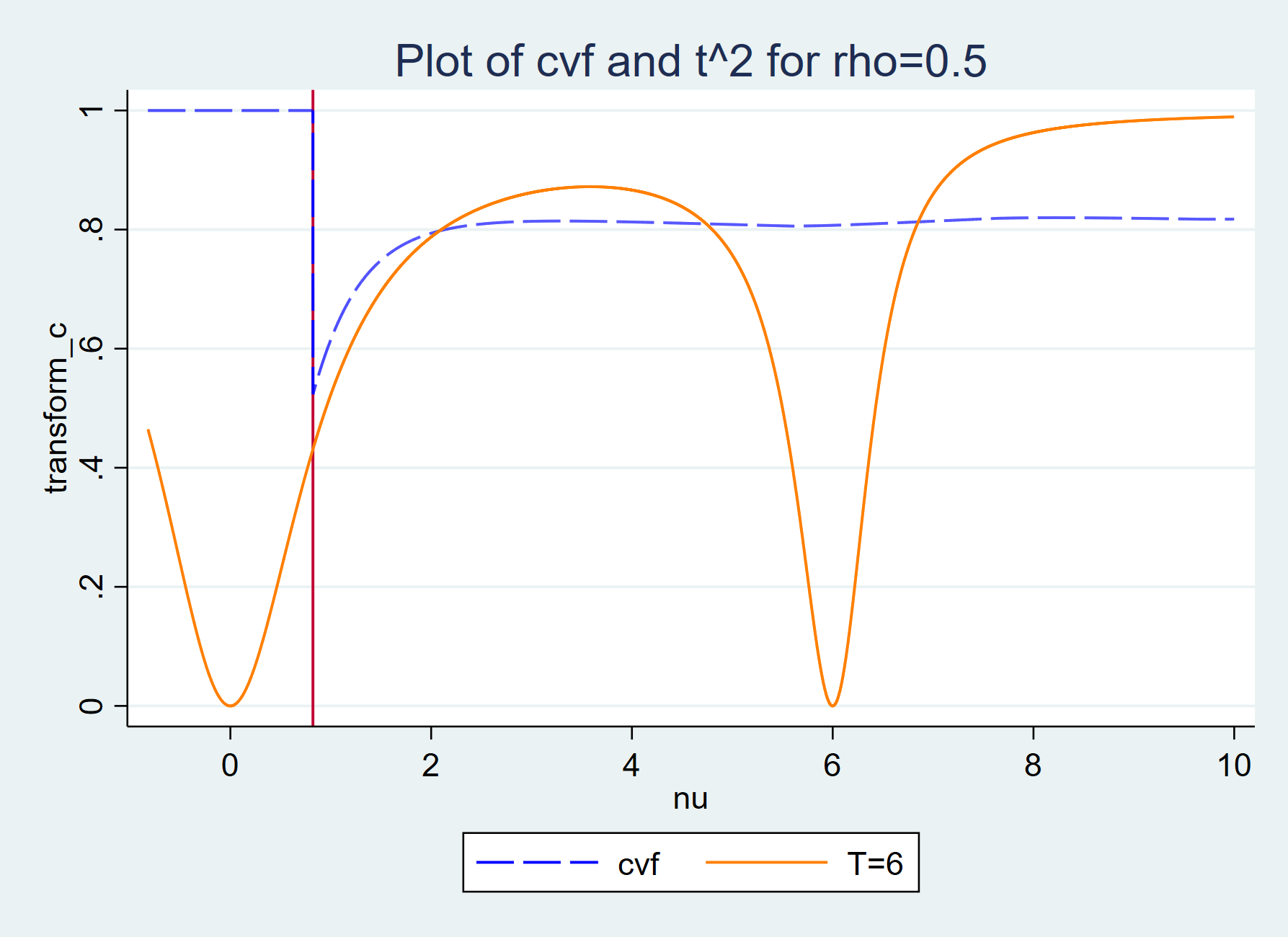} 
    \footnotesize\flushleft The dashed line corresponds to $c(\nu,\rho)$ while the solid line corresponds to $t^2(\nu,T,\rho)$. The $y$ axis has been transformed using $y^* = \frac{y}{1+y}$, so $\infty$ corresponds to 1.
\end{figure}

\begin{figure}
    \centering
    \caption{Curves with a tangent} \label{fig:tancross}
    \includegraphics[scale=.2]{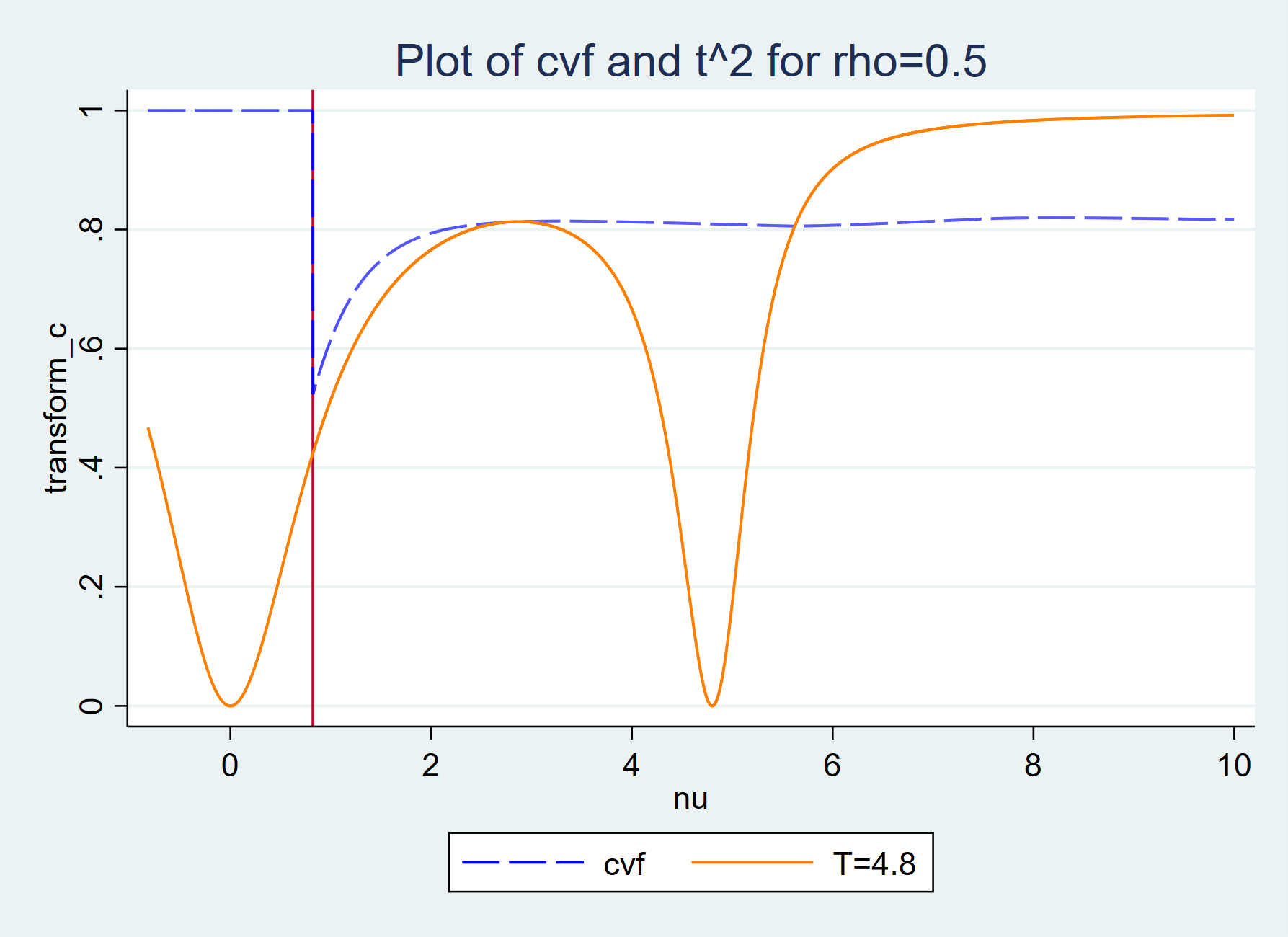} 
    \footnotesize\flushleft The dashed line corresponds to $c(\nu,\rho)$ while the solid line corresponds to $t^2(\nu,T,\rho)$. The $y$ axis has been transformed using $y^* = \frac{y}{1+y}$, so $\infty$ corresponds to 1.
\end{figure}

\section{Properties}
\subsection{Theoretical Power Bounds}
To compare VtF and VtFo in the many instruments environment with existing valid procedures in the literature, I first compare their power theoretically, then compare their numerical performance in the next section. Since the LM statistic of \citet{matsushita2022jackknife} $\xi(\beta_0)$ is entirely analogous to $t_{AR}(\beta_0)$ in the just-identified environment, the relative performance between $\xi(\beta_0)$ relative to VtF here is identical to the relative performance between $t_{AR}(\beta_0)$ and VtF in the just-identified case of LMMPY, so the analysis is omitted for brevity. The $AR(\beta_0)$ statistic of \citet{mikusheva2022inference} has no immediate analog in the just-identified case, so its analysis is warranted.

A useful benchmark of power is the power bound, which is the rejection probability as $\Delta \rightarrow \infty$. Due to the result in \citet{andrews2019optimal}, the power bound is one minus the probability that we get an unbounded confidence set (CS). Hence, to characterize the power bound, it suffices to characterize the probability that the various procedures yield an unbounded CS for a given DGP. The rest of this subsection will show that the VtF procedure and the AR procedure proposed by \citet{mikusheva2022inference} have the same power bound. 

First, consider the two-sided procedures. Based on the just-identified theory of LMMPY, $t^2$ yields an unbounded CS if and only if their $f$ satisfies $f^2\leq 3.84$. Applying their result to our context, using the VtF for $t_{JIVE}$ yields an unbounded CS when $\nu^2 \leq 3.84$. Turning to the \citet{mikusheva2022inference} statistic, we can observe that, when using their $\hat{\Phi}(\beta_0)$,

\begin{align*}
    AR(\beta_0) &= \frac{Q_{e\left(\beta_{0}\right)e\left(\beta_{0}\right)}}{\sqrt{\hat{\Phi}\left(\beta_{0}\right)}} \\
    &= \frac{Q_{YY}-2Q_{XY}\beta_{0}+Q_{XX}\beta_{0}^{2}}{\sqrt{B_{YYYY}-4B_{YYYX}\beta_{0}-2B_{YYXX}\beta_{0}^{2}+4B_{YXYX}\beta_{0}^{2}+4B_{YXXX}\beta_{0}^{3}+B_{XXXX}\beta_{0}^{4}}} 
\end{align*}

where $B_{XYXY}$ is defined as:
\begin{align*}
    B_{XYXY}=\frac{2}{K}\sum_{i}\sum_{j\ne i}\frac{P_{ij}^{2}}{M_{ii}M_{jj}+M_{ij}^{2}}\left[X_{i}M_{i}Y\right]\left[X_{j}M_{j}Y\right]
\end{align*}

Since $AR(\beta_0)$ is normally distributed, if we were to use the rejection rule of $AR(\beta_0)^2>3.84$, then observe that the non-rejection region of $\beta_0$ that satisfies $AR(\beta_0)^2 \leq 3.84$ can be written as a quartic inequality in $\beta_0$ with the leading term $Q_{XX}^2 - 3.84 B_{XXXX}$. Hence, the CS is unbounded when $Q_{XX}^2 \leq 3.84 B_{XXXX}$. Since $B_{XXXX} \xrightarrow{p} \Upsilon$, this condition is equivalent to $\nu^2 \leq 3.84$, which is identical to VtF. 

Next, consider the one-sided procedures. The acceptance rule for VtFo is $t_{JIVE}\leq c\left(\nu;\rho\right)$. When looking at the power bound, we consider the case when $\rho(\beta_0)\rightarrow\pm1$, where the confidence set is the whole real line when $\nu\leq1.645$, because $c(.)\rightarrow\infty$. Hence, the power bound is one minus the probability that $\nu\leq1.645$. 

The (one-sided) \citet{mikusheva2022inference} acceptance rule is $AR(\beta_0) \leq 1.645$. This condition can be expanded to be:
\small
\[
Q_{YY}-2Q_{XY}\beta_{0}+Q_{XX}\beta_{0}^{2}\leq1.645\sqrt{B_{YYYY}-4B_{YYYX}\beta_{0}-2B_{YYXX}\beta_{0}^{2}+4B_{YXYX}\beta_{0}^{2}+4B_{YXXX}\beta_{0}^{3}+B_{XXXX}\beta_{0}^{4}}
\]

\normalsize
To save on notation, the above inequality can be written as $LHS\leq1.645\left(RHS\right)^{1/2}$. Since RHS is positive, we cannot reject the null whenever $LHS\leq0$. A realization of data where $Q_{XX}\leq0$ (which occurs if and only if $\nu\leq0$) is sufficient for the CS to be unbounded. To see this, when $Q_{XX}\leq0$, there exists $\beta_{0}^L$ such that for all $\beta_0: \beta_0 \leq \beta_0^L$, we have $LHS\leq0$. If $LHS>0$, then we can innocuously square both sides of the inequality, and the leading quartic term is $L_4 = 1.645^{2}B_{XXXX}-Q_{XX}^{2}$. If $L_{4}<0$, then we have an unbounded set of $\beta_{0}$ that satisfy the inequality. But in this unbounded set of $\beta_{0}$, we can only accept $\beta_{0}$'s such that $LHS\geq0$. Thus, the  condition for having an unbounded set is $1.645^{2}B_{XXXX}-Q_{XX}^{2}\leq0$, which translates to $\nu\leq1.645$. This condition is identical to VtFo. 

\subsection{Power Curves}

To calculate asymptotic power curves, I draw the asymptotic $Q$'s to calculate power. Since $Q$ are asymptotic statistics, this exercise
abstracts from drawing a dataset $(X,Y,Z)$. I draw the transformed (TR) object directly, and treat $\Upsilon,\Psi,\Phi,\tau$ as known. 

\[
Q_{TR}:=\left[\begin{array}{c}
AR_{TR}\\
\xi_{TR}\\
\nu
\end{array}\right]=\left[\begin{array}{c}
Q_{ee}/\sqrt{\Upsilon\Phi}\\
Q_{Xe}/\sqrt{\Upsilon\Psi}\\
\nu
\end{array}\right]\sim N\left(\left[\begin{array}{c}
0\\
0\\
S
\end{array}\right],\left[\begin{array}{ccc}
1 & r/2 & r^{2}\\
r/2 & \frac{1}{2}\left(1+r^{2}\right) & r/2\\
r^{2} & r/2 & 1
\end{array}\right]\right)
\]

Due to the formulae in the setting for the variance objects under the alternative (i.e., $\tau\left(\beta_{0}\right),\Phi\left(\beta_{0}\right)$, etc.), for any given $\Delta$, these alternative variances can be calculated. For every draw of $Q_{TR}$, and a given $\Delta$, I can calculate the \citet{mikusheva2022inference} test statistic $AR$, the \citet{matsushita2022jackknife} test statistic $\xi$, correlation $\rho\left(\beta_{0}\right)$, and assemble the $t_{JIVE}^{2}$. These objects are sufficient to assess whether each inference procedure rejects the null or not. For every $\Delta,S,r$, I take 10,000 draws, and plot power curves to display the results. Since the distribution is known, we can analytically calculate the power bounds, which are also displayed in the plot. 

\begin{figure}
    \centering
    \caption{Power Curves comparing various procedures} \label{fig:power}
    \includegraphics[scale=1]{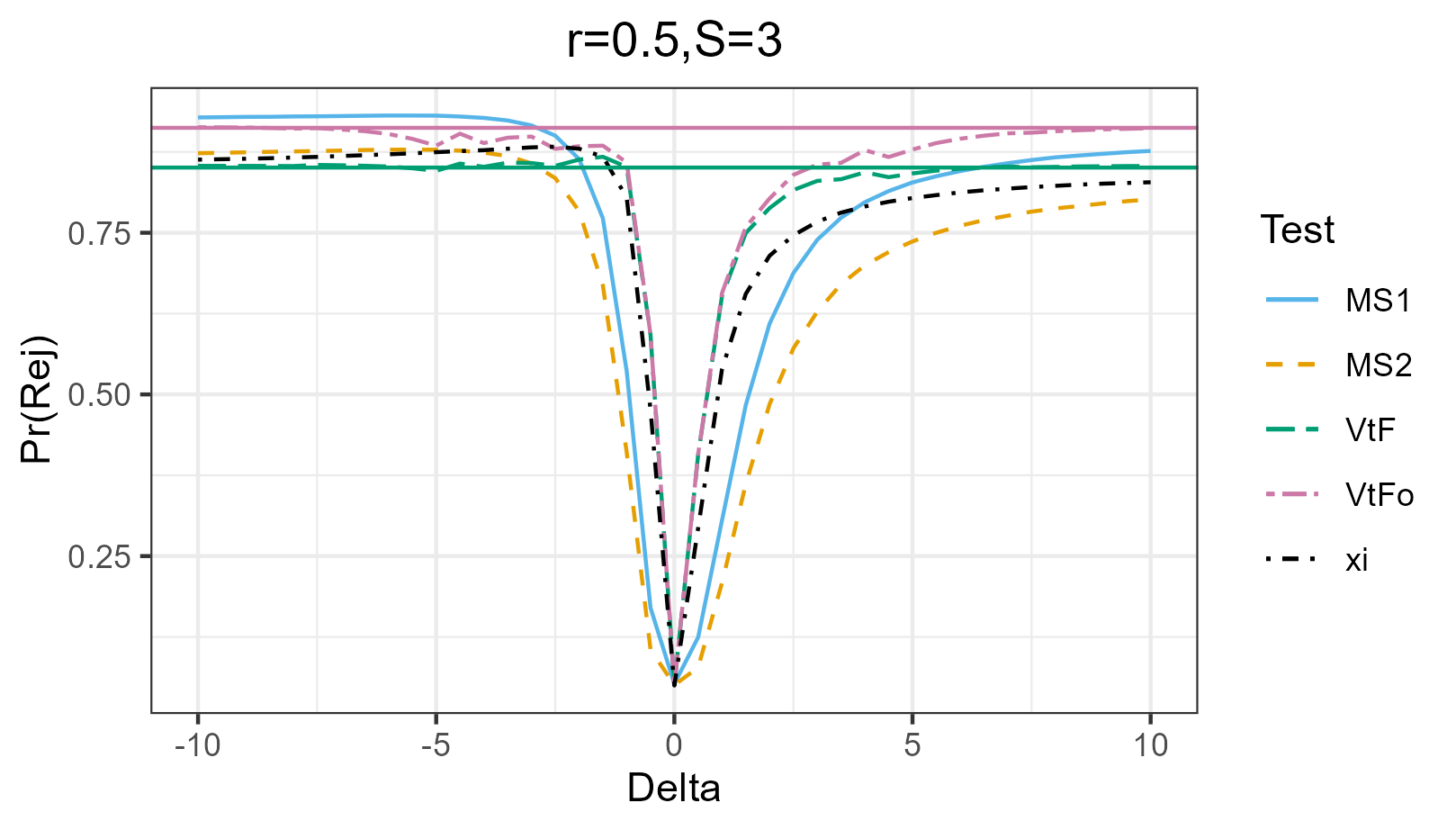} 
    \footnotesize\flushleft $\Delta := \beta - \beta_0$. MS1 and MS2 denote the one-sided and two-sided procedures of \citet{mikusheva2022inference} respectively. xi uses the $\xi(\beta_0)$ statistic of \citet{matsushita2022jackknife}. The horizontal lines denote the power bounds as $\Delta \rightarrow \pm \infty$.
\end{figure}

\Cref{fig:power} displays a power curve. In the data-generating process (DGP) with $r=0.5, S=3$, it is evident that all procedures have the correct size of 0.05 at $\Delta =0$. The performance of xi relative to VtF is identical to that of LMMPY in the two-sided procedures. While MS2 is somewhat less powerful than xi, such a comparison is not fruitful as \citet{mikusheva2022inference} proposed the one-sided ``MS1" procedure instead. Comparing MS1 and VtFo, observe that VtFo is more powerful than MS1 on one side of the alternative than the other. As we extend $\Delta$ out to larger values, both VtFo and MS1 converge to the power bound that was calculated theoretically. Since $S>0$, the power bound of the one-sided procedures is above the power bound of two-sided procedures.

This paper builds on the existing literature on MWIV by using its canonical asymptotic environment. By observing that the distribution of $t_{JIVE}$ is identical to $t_{TSLS}$ in the environment with just-identified instruments, this paper proposes the analog of the VtF procedure from LMMPY adapted to many instruments, and adapts the procedure to accommodate environments where the sign of the expected first-stage effect is known. It is further observed that there is an analogous conditional Wald procedure that is valid. The power properties of the proposed VtF procedure are similar to those in LMMPY, suggesting that the procedure is comparable to those in the existing literature. 


\appendix
\section{Proofs}

\begin{proof} [Proof of Lemma \ref{lem:JIVE}]
First observe that

\begin{align*}
\left(\hat{\beta}_{JIVE}-\beta_{0}\right) & =\frac{\sum_{i}\sum_{j\ne i}P_{ij}Y_{i}X_{j}}{\sum_{i}\sum_{j\ne i}P_{ij}X_{i}X_{j}}-\beta_{0}\\
 & =\frac{\sum_{i}\sum_{j\ne i}P_{ij}\left(X_{i}\beta_{0}+e_{i}\left(\beta_{0}\right)\right)X_{j}}{\sqrt{K}Q_{XX}}-\beta_{0}\\
 & =\frac{\beta_{0}\sqrt{K}Q_{XX}+\sum_{i}\sum_{j\ne i}P_{ij}e_{i}\left(\beta_{0}\right)X_{j}-\beta_{0}\sqrt{K}Q_{XX}}{\sqrt{K}Q_{XX}} =\frac{Q_{Xe\left(\beta_{0}\right)}}{Q_{XX}}
\end{align*}

With

\[
\hat{V}=\frac{\sum_{i}\left(\sum_{j\ne i}P_{ij}X_{j}\right)^{2}\frac{\hat{e}_{i}M_{i}\hat{e}}{M_{ii}}+\sum_{i}\sum_{j\ne i}\tilde{P}_{ij}^{2}M_{i}X\hat{e}_{i}M_{j}X\hat{e}_{j}}{\left(\sum_{i}\sum_{j\ne i}P_{ij}X_{i}X_{j}\right)^{2}}
\]

we know numerically that:

\begin{align*}
\sum_{i}\left(\sum_{j\ne i}P_{ij}X_{j}\right)^{2}\frac{\hat{e}_{i}M_{i}\hat{e}}{M_{ii}} & =\sum_{i}\left(\sum_{j\ne i}P_{ij}X_{j}\right)^{2}\frac{\left(e_{i}\left(\beta_{0}\right)-X_{i}\frac{Q_{Xe\left(\beta_{0}\right)}}{Q_{XX}}\right)M_{i}\left(e\left(\beta_{0}\right)-X\frac{Q_{Xe\left(\beta_{0}\right)}}{Q_{XX}}\right)}{M_{ii}}\\
\sum_{i}\sum_{j\ne i}\tilde{P}_{ij}^{2}M_{i}X\hat{e}_{i}M_{j}X\hat{e}_{j} & =\sum_{i}\sum_{j\ne i}\tilde{P}_{ij}^{2}M_{i}X\left(e_{i}\left(\beta_{0}\right)-X_{i}\frac{Q_{Xe\left(\beta_{0}\right)}}{Q_{XX}}\right)M_{j}X\left(e_{j}\left(\beta_{0}\right)-X_{j}\frac{Q_{Xe\left(\beta_{0}\right)}}{Q_{XX}}\right)
\end{align*}

Then, by expanding $\hat{V}$,

\begin{align*}
\hat{V}= & \frac{1}{KQ_{XX}^{2}}\left\{ \sum_{i}\left(\sum_{j\ne i}P_{ij}X_{j}\right)^{2}\frac{e_{i}\left(\beta_{0}\right)M_{i}e\left(\beta_{0}\right)}{M_{ii}}+\sum_{i}\sum_{j\ne i}\tilde{P}_{ij}^{2}M_{i}Xe_{i}\left(\beta_{0}\right)M_{j}Xe_{j}\left(\beta_{0}\right)\right\} \\
 & -\frac{Q_{Xe\left(\beta_{0}\right)}}{KQ_{XX}^{3}}\left\{ \sum_{i}\left(\sum_{j\ne i}P_{ij}X_{j}\right)^{2}\left(\frac{e_{i}\left(\beta_{0}\right)M_{i}X}{M_{ii}}+\frac{X_{i}M_{i}e_{i}\left(\beta_{0}\right)}{M_{ii}}\right)+2\sum_{i}\sum_{j\ne i}\tilde{P}_{ij}^{2}M_{i}Xe_{i}\left(\beta_{0}\right)M_{j}XX_{j}\right\} \\
 & +\frac{Q_{Xe\left(\beta_{0}\right)}^{2}}{KQ_{XX}^{4}}\left\{ \sum_{i}\left(\sum_{j\ne i}P_{ij}X_{j}\right)^{2}\frac{X_{i}M_{i}X}{M_{ii}}+\sum_{i}\sum_{j\ne i}\tilde{P}_{ij}^{2}M_{i}XX_{i}M_{j}XX_{j}\right\} 
\end{align*}

By substituting these expressions with the hat variance and covariance
objects, we get:

\[
\hat{t}_{JIVE}^2=(\beta_0) = \frac{Q_{Xe\left(\beta_{0}\right)}^{2}}{\hat{\Psi}\left(\beta_{0}\right)-2\frac{Q_{Xe(\beta_{0})}}{Q_{XX}}\hat{\tau}\left(\beta_{0}\right)+\frac{Q_{Xe(\beta_{0})}^{2}}{Q_{XX}^{2}}\hat{\Upsilon}}
\]
Then, with the appropriate normalizations, the numerical equivalence is immediate. 

Under Assumption \ref{asmp:var_consistent},
\begin{align*}
\hat{t}_{JIVE}^2 &= \frac{Q_{Xe\left(\beta_{0}\right)}^{2} (1+o_P(1))}{\Psi\left(\beta_{0}\right)-2\frac{Q_{Xe(\beta_{0})}}{Q_{XX}}\tau\left(\beta_{0}\right)+\frac{Q_{Xe(\beta_{0})}^{2}}{Q_{XX}^{2}}\Upsilon} \\
&= \frac{ \frac{Q_{Xe\left(\beta_{0}\right)}^{2}}{\Psi(\beta_0)} (1+o_P(1))}{1-2\frac{Q_{Xe(\beta_{0})/\sqrt{\Psi(\beta_0)}}}{Q_{XX}/\sqrt{\Upsilon}} \frac{\tau\left(\beta_{0}\right)}{\sqrt{\Psi(\beta_0)/ \sqrt{\Upsilon}}}+\frac{Q_{Xe(\beta_{0})}^{2}/\Psi(\beta_0)}{Q_{XX}^{2}\Upsilon^2}} \\
&=\frac{\xi\left(\beta_{0}\right)^{2}(1+o_{P}(1))}{1-2\frac{\xi\left(\beta_{0}\right)}{\nu}\rho\left(\beta_{0}\right)+\frac{\xi\left(\beta_{0}\right)^{2}}{\nu^{2}}}
\end{align*}
\end{proof}

\bibliographystyle{ecta}
\bibliography{mwiv}

\end{document}